\documentclass[journal]{IEEEtran}

\usepackage{amsfonts}

\usepackage{amssymb}
\usepackage{amsmath}
\usepackage{mathrsfs}
\usepackage{graphics}
\usepackage{epsfig}
\usepackage{epstopdf}
\usepackage{amsthm}
\usepackage{textcomp}
\usepackage{color}
\usepackage{float}
\usepackage{multirow}
\usepackage{diagbox}

\newtheorem{rk}{Remark}
\newtheorem{thm}{Theorem}
\newtheorem{lm}{Lemma}

\newtheorem{cor}{Corollary}
\newtheorem{de}{Definition}

\ifCLASSINFOpdf
\else
\fi
\usepackage{algorithmic}
\usepackage{algorithm}

\usepackage{array}
\begin{document}
%
\title{Distributed interaction between computer virus and patch: A modeling study}
%
%
%

\author{Lu-Xing Yang,~\IEEEmembership{Member,~IEEE,}
        Xiaofan Yang,~\IEEEmembership{Member,~IEEE,}
        Qingyi Zhu, and
        Chenquan Gan
        
\thanks{L.-X. Yang is with the School of Information Technology, Deakin University, VIC 3125, Australia e-mail: ylx910920@gmail.com.}
\thanks{X. Yang is with the School of Software Engineering, Chongqing University, 400044, P.R. China e-mail: xfyang1964@gmail.com.}
\thanks{Q. Zhu and C. Gan are with School of Cyber Security and Information Law, Chongqing University of Posts and Telecommunications,
Chongqing 400065, P.R. China
	e-mails: zhuqy@cqupt.edu.cn, gcq2010cqu@163.com}
}

\maketitle

\begin{abstract}
	
	The decentralized patch distribution mechanism holds significant promise as an alternative to its centralized counterpart. For the purpose of accurately evaluating the performance of the decentralized patch distribution mechanism and based on the exact SIPS model that accurately captures the average dynamics of the interaction between viruses and patches, a new virus-patch interacting model, which is known as the generic SIPS model, is proposed. This model subsumes the linear SIPS model. The dynamics of the generic SIPS model is studied comprehensively. In particular, a set of criteria for the final extinction or/and long-term survival of viruses or/and patches are presented. Some conditions for the linear SIPS model to accurately capture the average dynamics of the virus-patch interaction are empirically found. As a consequence, the linear SIPS model can be adopted as a standard model for assessing the performance of the distributed patch distribution mechanism, provided the proper conditions are satisfied.  

\end{abstract}

\begin{IEEEkeywords}
	networks, epidemiology, stochastic processes, nonlinear systems, stability
\end{IEEEkeywords}

%
\IEEEpeerreviewmaketitle

\section{Introduction}
%
%
%
%
It has long been a hotspot of research in the cybersecurity field to improve the defensive capability of computer networks against digital viruses \cite{Obaidat2007}. Some fast-spreading viruses, such as worms, could infect a significant fraction of hosts in the Internet within only a matter of minutes \cite{ZouC2007, WenS2013, WangY2014, WenS2014}. With the proliferation of mobile devices such as smart phones and tablet PCs, the threat posed by electronic viruses is becoming increasingly serious \cite{WangP2009, GaoC2013, PengS2014}. Consequently, a rapid response mechanism to the sprawl of computer infections is in demand.

As the natural enemy of digital viruses, virus patches are developed to detect and clean up the viruses stayed in each single infected host. To put a new patch to use, it must be copied and rapidly distributed to all hosts in the network. The conventional centralized patch distribution mechanism, which requires all clients to download patches from the server, suffers from a serious bottleneck at the server end \cite{Grimes2001, Kara2001}. Furthermore, the bottleneck cannot be well relieved by adding more servers/caches \cite{Grimes2001, Kara2001}. Additionally, the centralized patch distribution scheme is vulnerable to the denial-of-service (DoS) attack, which could be launched by virus writers in order to delay patch distribution. 

To circumvent the bottleneck brought about by the centralized patch distribution mechanism, in the early 21st century, Grimes \cite{Grimes2001} and Kara \cite{Kara2001} suggested the possibility of distributing patches in a decentralized way. Later, Toyoizumi and Kara \cite{Toyoizumi2002} made a preliminary analysis of the performance of the decentralized patch distribution approach. Gupta and DuVarney \cite{Gupta2004} empirically found that the decentralized distribution mechanism is effective, without suffering from bottlenecks or causing network congestion. Thereby they declared that the decentralized patch distribution mechanism holds significant promise as an alternative to its centralized counterpart. Liljenstam and Nicol \cite{Liljenstam2004} noticed that, before the decentralized patch distribution scheme is applicable to real scenarios, some legal, ethical and technical issues must be well addressed. Based on a detailed analysis, Weaver and Ellis \cite{Weaver2006} convincingly declared that the decentralized distributed scheme must be supported by the patch-management systems deployed on the hosts rather than exploiting vulnerabilities of the hosts. These pioneer researches promoted an extensive modeling study of the decentralized patch distribution approach \cite{Tamimi2006a, Tanachaiwiwat2009, WangFW2013, YangLX2013, Mishra2014, YangLX2015a}. However, the models proposed in these references have two common defects: (a) the influence of the structure of the underlying network is ignored partly or even completely, (b) the linearly growing infecting/patching rates overestimate the actual rates. As a result, the resulting assessments of the distributed patch distribution mechanism are highly doubtful. 

By applying the  continuous-time Markov chain technique to the basic assumptions about the Susceptible-Infected-Susceptible (SIS) virus, Van Mieghem et. al. \cite{Mieghem2009} formulated an epidemic model known as the exact SIS model, because it accurately captures the average dynamics of the SIS virus. In particular, the model takes into full consideration the impact of the structure of the network on the virus prevalence. Due to extremely high dimensionality, the exact SIS model is not mathematically tractable. To overcome this defect, Van Mieghem et. al. \cite{Mieghem2009} suggested a node-level SIS model with linear infecting rate (the node-level linear SIS model) as an approximation of the exact SIS model. This model is much easier to analyze than the exact SIS model, at the cost of accuracy \cite{Mieghem2011}. Subsequently, a number of node-level epidemic models have been suggested \cite{Guo2012, Sahneh2012a, Sahneh2013, Shuai2013, Sahneh2014, YangLX2015b, Santos2015, Selley2015, LiuJ2016, YangLX2016a}. For the purpose of more accurately assessing the decentralized patch distribution approach, recently Yang and Yang \cite{YangLX2016b} proposed a node-level virus-patch interacting model. 

All of the above node-level epidemic models are established based on the assumption that the infecting rates are all linear in the virus occupation probabilities of the nodes in the population. However, the real infecting rates would flatten out rather than growing linearly. As a result, there might be a significant gap between the predicted dynamics and the actual one. To enhance the accuracy of these epidemic models, much effort has been taken. Van Mieghem et. al. \cite{Mieghem2009} empirically found that when the effective infection rate is lower than the epidemic threshold, the node-level SIS model fits well with the exact SIS model. Cator and Van Mieghem \cite{Cator2012} introduced a much more complex approximation of the exact SIS model. Van Mieghem and van de Bovenkamp \cite{Mieghem2015} theoretically estimated the accuracy of the node-level SIS model. 

A feasible approach to the enhancement of the accuracy of node-level epidemic models with linear infecting rates is to replace the linear infection rates with nonlinear ones. Inspired by the active cyber defense proposed by Xu \cite{XuSH2014a}, a number of node-level epidemic models with nonlinear infecting/curing rates, which can be viewed as variants of the node-level SIS model and can be used to assess the performance of the decentralized patch distribution mechanism, have recently been proposed \cite{LuWL2013, XuSH2015, ZhengR2015}. Nevertheless, the results on node-level epidemic models with nonlinear infecting rates are still rare.

For the purpose of more accurately assessing the performance of the decentralized patch dissemination mechanism, in this paper we propose and study a generic virus-patch interacting model. Our main contributions are sketched as follows.

\begin{enumerate}
	
	\item[(a)] By applying the continuous-time Markov chain technique to a basic set of assumptions, a high-dimensional dynamic model, which is known as the exact SIPS model bacause it accurately captures the average interacting dynamics between viruses and patches, is formulated. This model forms the foundation of this work. An equivalent form of the exact SIPS model is derived. Based on the added independence assumption, the exact SIPS model degenerates to a node-level epidemic model with linear infecting/patching rates, which is known as the linear SIPS model. The linear infecting/patching rates in the model more or less overestimate the real infecting/patching rates.
	
	\item[(b)] To more accurately predict the virus-patch interaction, a node-level epidemic model, which is known as the generic SIPS model because the infecting/patching rates satisfy a generic set of conditions, is proposed. The dynamics of the model is studied comprehensively. In particular, a criterion for the final extinction of both viruses and patches, a criterion for the long-term survival of viruses and the final extinction of patches, a criterion for the final extinction of viruses and the long-term survival of patches, and a criterion for the long-term survival of both viruses and patches are presented, respectively. 
	
	\item[(c)]  By means of extensive experiments, some conditions for the linear SIPS model to satisfactorily capture the dynamics of the virus-patch interactions are found. As a consequence, the linear SIPS model can be adopted as a standard model for assessing the performance of the distributed patch distribution mechanism, provided the conditions are satisfied.
	
\end{enumerate}

The subsequent materials of this work are organized as follows. Section 2 formulates the exact SIPS model, derives its equivalent form, describes the linear SIPS model, and presents the generic SIPS model. Section 3 reveals the dynamical properties of the generic SIPS model. Some conditions for the linear SIPS model to satisfactorily capture the dynamics of the virus-patch interactions are found in Section 4. Finally, Section 5 outlines this work and suggests some interesting topics of research.

\section{A generic virus-patch interacting model}

This section aims to establish a generic continuous-time dynamic model capturing the interaction between digital viruses and virus patches.

\subsection{Basic assumptions}

Consider a population of $N$ computing nodes labelled $1, 2, ..., N$. As with the traditional SIPS models \cite{YangLX2015a, YangLX2016b}, at any given time every node in the population is assumed to be in one of three possible states:  \emph{susceptible}, \emph{infected}, and \emph{patched}. A susceptible node is not infected with any virus but is vulnerable to the newest virus, because it hasn't yet acquired the newest patch. An infected node is infected with a virus. A patched node is not infected with any virus and is immune to the newest virus (as well as all the older viruses), because it has acquired the newest patch. Let $X_i(t)$ = 0, 1, and 2 denote that at time $t$ node $i$ is susceptible, infected, and patched, respectively. The state of the population at time $t$ is represented by the vector
$\mathbf{X}(t) = (X_1(t), X_2(t), \cdots, X_N(t))^T$.

Let us impose a basic set of assumptions on the rates of state transitions of nodes as follows, where $1 \leq i, j \leq N$. 

\begin{enumerate}
	\item[(H$_1$)] Due to the infection of viruses, at any time a susceptible node $i$ is infected by an infected node $j$ at rate $\beta_{ij} \geq 0$, which is proportional to the probability with which a virus propagation channel from node $j$ to node $i$ is established, the rate at which node $j$ tries to deliver a virus copy to node $i$ when the virus propagation channel is present, and the probability with which node $i$ accepts a virus copy when it comes. Let $\mathbf{M}_{\beta}= \left[\beta_{ij}\right]_{N \times N}$.
	\item[(H$_2$)] Due to the decentralized distribution of the newest patch, at any time a susceptible node $i$ is patched by a patched node $j$ at rate $\delta_{ij}^{(1)} \geq 0$, which is proportional to the probability that a patch distribution channel from node $j$ to node $i$ is established, the rate at which node $j$ tries to deliver a copy of the newest patch to node $i$ when the patch distribution channel is present, and the probability that node $i$ accepts a copy of the newest patch when it comes. Let $\mathbf{M}_{\delta_1}= \left[\delta_{ij}^{(1)}\right]_{N \times N}$.
	\item[(H$_3$)] Due to the decentralized distribution of the newest patch, at any time an infected node $i$ is patched by a patched node $j$ at rate $\delta_{ij}^{(2)} \geq 0$, which is proportional to the probability that a patch distribution channel from node $j$ to node $i$ is established, the rate at which node $j$ tries to deliver a copy of the newest patch to node $i$ when the patch distribution channel is present, and the probability that node $i$ accepts a copy of the newest patch when it comes. Let $\mathbf{M}_{\delta_2}= \left[\delta_{ij}^{(2)}\right]_{N \times N}$. In what follows, it is always assumed that $\delta_{ij}^{(1)} > 0$ if and only if $\delta_{ij}^{(2)} > 0$.
	\item[(H$_4$)] Due to the system reinstallation, at any time an infected node $i$ becomes susceptible at rate $\gamma_i>0$, which is proportional to the probability that node $i$ is aware of a virus and the average rate at which the system in node $i$ is reinstalled when a virus is perceived. Let  $\mathbf{D}_{\gamma} = diag\left(\gamma_i\right)$.
	\item[(H$_5$)] Due to the patch failure caused by a new virus, a patched node $i$ becomes susceptible at rate $\alpha_i>0$, which is proportional to the rate at which a new virus emerges. Let $\mathbf{D}_{\alpha} = diag\left(\alpha_i\right)$.
\end{enumerate}

Let $V = \{1, 2, \cdots, N\}$. Define the \emph{virus-propagating network}, $G_v = (V, E_v)$, as follows: for any $i,j \in V$, $(j, i) \in E_v$ if and only if $\beta_{ij}> 0$. Define the \emph{patch-distributing network}, $G_p = (V, E_p)$, as follows: for any $i,j \in V$, $(j, i) \in E_p$ if and only if $\delta_{ij}^{(1)} > 0$ (equivalently, $\delta_{ij}^{(2)} > 0$). Then, viruses propagate through $G_v$, while patches are distributed through $G_p$. For technical reasons, we need the following additional assumption.

\begin{enumerate}
	\item[(H$_6$)] $G_v$ and $G_p$ are both strongly connected.
\end{enumerate}

We assume that all the forthcoming virus-patch interacting models comply with these six assumptions.

\begin{figure}[!t]
	\centering
	\includegraphics[width=2.5in]{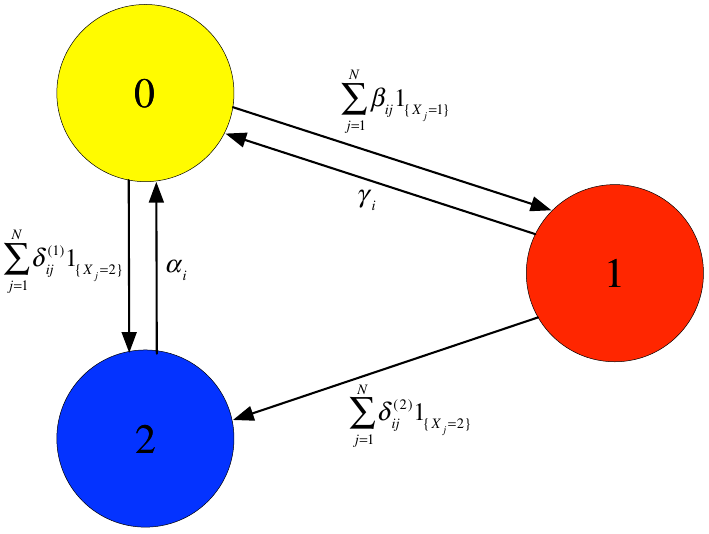}
	\caption{The diagram of state transitions under the exact SIPS model.}
\end{figure}

\subsection{The exact SIPS model}

For fundamental knowledge on continuous-time Markov chain, see Ref. \cite{Stewart2009}. Another way of representing the state of the population at time $t$ is by the decimal number $i(t)=\sum_{k=1}^{N}X_k(t)3^{k-1}$. In this context, there are totally $3^N$ possible population states: $0, 1, \cdots, 3^{N} - 1$. According to the basic assumptions, the infinitesimal generator $\mathbf{Q}=\left[q_{ij}\right]_{3^N \times 3^N}$ for the virus-patch interacting process is given as
\[
q_{ij}=\left\{
\begin{aligned}
\gamma_m    \quad\quad\quad&\quad \text{if} \quad i=j+3^{m-1},\\
& \quad m=1,2,\cdots,N, x_{m}=1 ;\\
\alpha_m  \quad\quad\quad& \quad \text{if} \quad i=j+2\cdot3^{m-1},\\
&  \quad m=1,2,\cdots,N, x_{m}=2; \\
\sum_{k=1}^{N}\beta_{mk}1_{\{x_{k}=1\}}   & \quad \text{if} \quad i=j-3^{m-1},\\
 & \quad m=1,2,\cdots,N, x_{m}=0;  \\
\sum_{k=1}^{N}\delta_{mk}^{(1)}1_{\{x_{k}=2\}}  & \quad \text{if} \quad i=j-2\cdot3^{m-1},\\
 & \quad m=1,2,\cdots,N, x_{m}=0; \\
\sum_{k=1}^{N}\delta_{mk}^{(2)}1_{\{x_{k}=2\}}  & \quad \text{if} \quad i=j-3^{m-1}, \\
 & \quad m=1,2,\cdots,N, x_{m}=1; \\
-\sum_{k=0,k\neq i}^{N-1}q_{ik}  & \quad \text{if}  \quad i=j;\\
 0  \quad \quad &\quad \text{otherwise}.
\end{aligned}
\right.
\]
where $i=\sum_{k=1}^{N}x_{k}3^{k-1}$, $1_A$ stands for the indicator function of set $A$. Let $s_i(t)$ denote the probability that at time $t$ the population is in state $i = \sum_{k=1}^{N}x_{k}3^{k-1}$. That is,
$s_{i}(t)=\Pr\left\{X_{1}(t)=x_{1},\cdots,X_{N}(t)=x_N\right\}$. Then $\mathbf{s}(t)=\left(s_{0}(t),\cdots,s_{3^{N}-1}(t)\right)^T$ obeys
\begin{equation}
	\frac{d\mathbf{s}^T(t)}{dt}=\mathbf{s}^T(t)\mathbf{Q}.
\end{equation}
This continuous-time Markov chain model accurately captures the average dynamics of the interaction between viruses and patches. Therefore, we refer to the model as the \emph{exact SIPS model}. Fig. 1 shows the diagram of state transitions under this model. Although the exact SIPS model is a linear differential system with the known solution $\mathbf{s}^{T}(t)=\mathbf{s}^{T}(0)e^{\mathbf{Q}t}$, its dimensionality grows exponentially with the size of the population, leading to its mathematical intractability.

\subsection{An equivalent form of the exact SIPS model}

Let $I_i(t) = \Pr\{X_i(t) = 1\}$, $P_i(t) = \Pr\{X_i(t) = 2\}$.

\begin{lm}
	The exact SIPS model is equivalent to the model
	
	\begin{equation}
		\left\{
		\begin{aligned}
			\frac{dI_i(t)}{dt} &= \sum_{j = 1}^N \beta_{ij} \Pr\{X_i(t) = 0, X_j(t) = 1\}\\
			&\quad -\sum_{j = 1}^N \delta_{ij}^{(2)} \Pr\{X_i(t) = 1, X_j(t) = 2\} - \gamma_i I_i(t), \\
			\frac{dP_i(t)}{dt} &= \sum_{j = 1}^N \delta_{ij}^{(1)} \Pr\{X_i(t) = 0, X_j(t) = 2\}\\
			&\quad +\sum_{j = 1}^N \delta_{ij}^{(2)} \Pr\{X_i(t) = 1, X_j(t) = 2\} - \alpha_i P_i(t),\\
			& \quad \quad i = 1, 2, \cdots, N.
		\end{aligned}
		\right.    
	\end{equation}
\end{lm}

\begin{IEEEproof}
	Given a sufficiently small time interval $\Delta t >0$, it follows from the total probability formula that, for $1 \leq i \leq N$,
	\[
	\begin{split}
	I_i(t+\Delta t)&= \left(1 - I_i(t) - P_i(t)\right) \Pr\{X_{i}(t+\Delta t)=1 \mid X_{i}(t)=0\} \\ \nonumber
	& + I_i(t) \Pr\{X_{i}(t+\Delta t)=1 \mid X_{i}(t)=1\} \\
	& + P_i(t) \Pr\{X_{i}(t+\Delta t)=1 \mid X_{i}(t)=2\}. \quad\quad (*)
	\end{split}
	\]
	By the conditional total probability formula and in view of model (1), we get that, for $1 \leq i \leq N$,
	\[
	\begin{split}
	&\Pr\{X_{i}(t+\Delta t)=1 \mid X_{i}(t)=0\} \\
	=&  \sum_{\mathbf{x} \in \{0, 1, 2\}^N, x_i = 0} \Pr\{X_{i}(t+\Delta t)=1 \mid X_{i}(t)=0, \mathbf{X}(t) = \mathbf{x}\}\\
	&\cdot \Pr\{\mathbf{X}(t) = \mathbf{x} \mid X_{i}(t)=0\} \\ \nonumber
	=& \frac{\Delta t}{1 - I_i(t) - P_i(t)} \sum_{\mathbf{x} \in \{0, 1, 2\}^N, x_i = 0} \sum_{j = 1}^N \beta_{ij} 1_{\{x_j=1\}}\Pr\{\mathbf{X}(t) = \mathbf{x}\} \\
	&+ o(\Delta t) \\ \nonumber
	\nonumber
	=& \frac{\Delta t}{1 - I_i(t) - P_i(t)}\sum_{j = 1}^N \beta_{ij} \Pr\{X_i(t) = 0, X_j(t) = 1\} + o(\Delta t). 
	\end{split}
	\]
	
	\noindent Similarly, we have, for $1 \leq i \leq N$,
	\[
	\begin{split}
	&\Pr\{X_{i}(t+\Delta t)=2 \mid X_{i}(t)=1\}\\
	=&\frac{\Delta t}{I_i(t)} \cdot \sum_{j = 1}^N \delta_{ij}^{(2)} \Pr\{X_i(t) = 1, X_j(t) = 2\} + o(\Delta t),
	\end{split}
	\]
	\[
	\begin{split}
	&\Pr\{X_{i}(t+\Delta t)=0 \mid X_{i}(t)=1\} = \gamma_i \Delta t+o(\Delta t).
	\end{split}
	\]
	
	\noindent It follows that, for $1 \leq i \leq N$,
	\[
	\begin{split}
	&\Pr\{X_{i}(t+\Delta t)=1 \mid X_{i}(t)=1\}\\
	=&1-\frac{\Delta t}{I_i(t)} \cdot \sum_{j = 1}^N \delta_{ij}^{(2)} \Pr\{X_i(t) = 1, X_j(t) = 2\} -\gamma_i \Delta t 
	+ o(\Delta t).
	\end{split}
	\]
	
	\noindent Besides, we have, for $1 \leq i \leq N$,
	\[
	\Pr\{X_{i}(t+\Delta t)=1 \mid X_{i}(t)=2\} = o(\Delta t).
	\]
	
	\noindent Substituting these equations into Eqs. (*), rearranging the terms, dividing both sides by $\Delta t$, and letting $\Delta t \rightarrow 0$, we get that, for $1 \leq i \leq N$,
	\[
	\begin{split}
	\frac{dI_i(t)}{dt}& = \sum_{j = 1}^N \beta_{ij} \Pr\{X_i(t) = 0, X_j(t) = 1\}\\
	&\quad-\sum_{j = 1}^N \delta_{ij}^{(2)} \Pr\{X_i(t) = 1, X_j(t) = 2\} - \gamma_i I_i(t).
	\end{split}
	\]
	
	\noindent Similarly, we can derive the last $N$ equations in Lemma 1.
\end{IEEEproof}

This equivalent model is not mathematically treatable, because it is not closed. If one attempted to close the equivalent model by adding joint probability terms, the resulting model would be of dimensionality $3^N$, which is still mathematically intractable.

\subsection{The linear SIPS model}

To simplify the exact SIPS model, let us make an added set of independence assumptions as follows, where $1 \leq i, j \leq N, i \neq j$. 

\begin{enumerate}
	\item[(H$_7$)] $\Pr\{X_i(t) = 0, X_j(t) = 1\} = (1 - I_i(t) - P_i(t))I_j(t).$
	\item[(H$_8$)] $\Pr\{X_i(t) = 0, X_j(t) = 2\} = (1 - I_i(t) - P_i(t))P_j(t).$
	\item[(H$_9$)] $\Pr\{X_i(t) = 1, X_j(t) = 2\} = I_i(t) P_j(t)$.
\end{enumerate}

\noindent Based on model (2) and assumption (H$_7$), we obtain the following model. 
\begin{equation}
	\left\{
	\begin{aligned}
		\frac{dI_i(t)}{dt} &= (1-I_i(t)-P_i(t))\sum_{j = 1}^N \beta_{ij} I_j(t)-I_i(t)\sum_{j = 1}^N \delta_{ij}^{(2)} P_j(t)\\
		&\quad- \gamma_i I_i(t), \\
		\frac{dP_i(t)}{dt} &= (1-I_i(t)-P_i(t))\sum_{j = 1}^N \delta_{ij}^{(1)} P_j(t)+I_i(t)\sum_{j = 1}^N \delta_{ij}^{(2)} P_j(t)\\
		&\quad- \alpha_i P_i(t), \\
		&\quad i = 1, 2, \cdots, N.
	\end{aligned}
	\right.    
\end{equation}

\noindent We refer to this model as the \emph{linear SIPS model}, because (a) the infecting rates, $\sum_{j = 1}^N \beta_{ij} I_j(t)$, are linear in $I_1(t), \cdots, I_N(t)$, and (b) the patching rates, $\sum_{j = 1}^N \delta_{ij}^{(1)} P_j(t)$ and $\sum_{j = 1}^N \delta_{ij}^{(2)} P_j(t)$, are linear in $P_1(t), \cdots, P_N(t)$. Fig. 2 shows the diagram of state transitions under the linear SIPS model.

\begin{figure}[!t]
	\centering
	\includegraphics[width=2.5in]{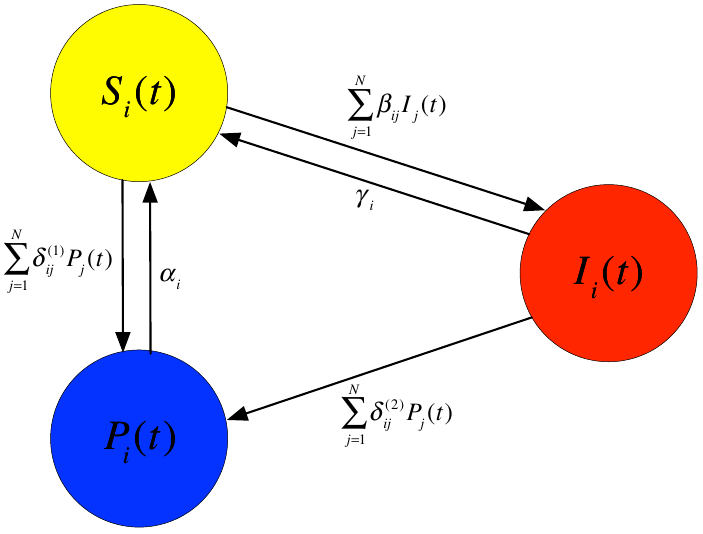}
	\caption{The diagram of state transitions under the linear SIPS model.}
\end{figure}

The linear SIPS model is mathematically tractable. As there is a difference between the linear infecting/patching rates and the actual infecting/patching rates, the dynamics of the model might deviate from the actual average dynamics of virus-patch interaction.

\subsection{The generic SIPS model}

Now, let us consider a more general virus-patch interacting model as follows.
\begin{equation}
	\left\{
	\begin{aligned}
		\frac{dI_i(t)}{dt}&=\left(1-I_i(t)-P_i(t)\right)f_i\left(I_1(t), \cdots, I_N(t)\right)\\
		&\quad -I_i(t)h_i\left(P_1(t), \cdots, P_N(t)\right)-\gamma_iI_i(t),\\
		\frac{dP_i(t)}{dt}&=\left(1-I_i(t)-P_i(t)\right)g_i\left(P_1(t), \cdots, P_N(t)\right) \\
		&\quad+I_i(t)h_i\left(P_1(t), \cdots, P_N(t)\right)-\alpha_iP_i(t), \\
		&\quad \quad i=1,2,\cdots, N.
	\end{aligned}
	\right.
\end{equation}
Here, the infecting/patching rates are assumed to satisfy the following generic conditions, where $1 \leq i, j, k \leq N$.

\begin{enumerate}
	
	\item[(C$_1$)] (Proximity) A susceptible node can and can only be infected by those infected nodes that are its out-neighbors in the virus-propagating network. That is, $f_i$ is dependent upon $x_j$ if and only if $\beta_{ij}> 0$. Likewise, an unpatched node can and can only be patched by those patched nodes that are its out-neighbors in the patch-distributing network. That is, $g_i$ is dependent upon $x_j$ if and only if $\delta_{ij}^{(1)}> 0$, and $h_i$ is dependent upon $x_j$ if and only if $\delta_{ij}^{(2)}> 0$.  
	
	\item[(C$_2$)] (Nullity) A virus cannot propagate unless there is already an infected node in the population. That is, $f_i(0, \cdots, 0)=0$. Likewise, a patch cannot be distributed unless there is already a patched node in the population. That is, $g_i(0, \cdots, 0)=h_i(0, \cdots, 0)=0$.
	
	\item[(C$_3$)] (Smoothness) The infecting/patching rates are sufficiently smooth. Technically speaking, $f_i$, $g_i$ and $h_i$ are twice continuously differentible.
	
	\item[(C$_4$)] (Monotonicity) The infecting/patching rates are strictly increasing. That is, $\frac{\partial f_i(x_1, \cdots, x_N)}{\partial x_j} > 0$ if $f_i$ is dependent upon $x_j$, 
	$\frac{\partial g_i(x_1, \cdots, x_N)}{\partial x_j} > 0$ if $g_i$ is dependent upon $x_j$, and $\frac{\partial h_i(x_1, \cdots, x_N)}{\partial x_j} > 0$ if $h_i$ is dependent upon $x_j$. 
	
	\item[(C$_5$)] (Concavity) The infecting/patching rates flatten out and tend to saturation. Specifically, $\frac{\partial^2 f_i(x_1, \cdots, x_N)}{\partial x_j\partial x_k}\leq 0$, $\frac{\partial^2 g_i(x_1, \cdots, x_N)}{\partial x_j\partial x_k}\leq 0$, and  $\frac{\partial^2 h_i(x_1, \cdots, x_N)}{\partial x_j\partial x_k}\leq 0$.
	
\end{enumerate}

\noindent We refer to model (4) as the \emph{generic SIPS model}. Obviously, this model subsumes the linear SIPS model as well as many SIPS models with nonlinear infecting/patching rates. Fig. 3 shows the diagram of state transitions under this model. 

\begin{figure}[!t]
	\centering
	\includegraphics[width=2.5in]{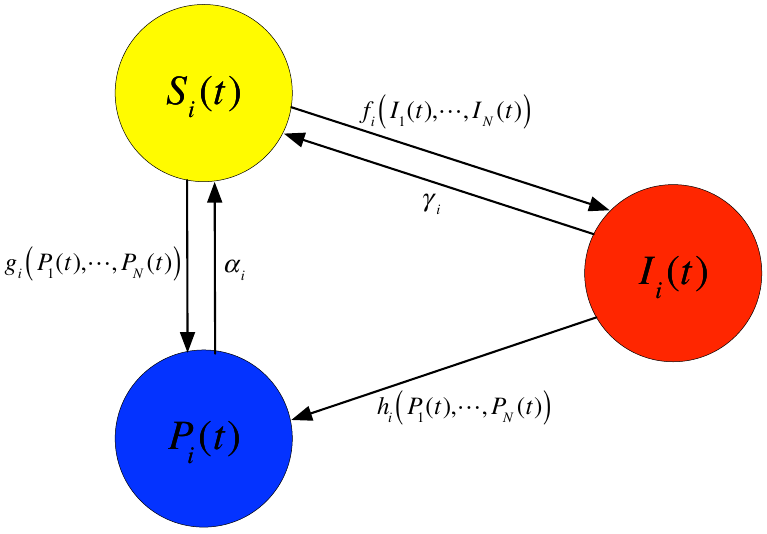}
	\caption{The diagram of state transitions under the generic SIPS model.}
\end{figure}

Let
\[
\Omega=\left\{(x_1,\cdots,x_{2N})\in \mathbb{R}_+^{2N}\mid x_i+x_{N+i}\leq 1, 1 \leq i \leq N\right\}.
\]
The initial state of model (4) lies in $\Omega$. It is easily shown that $\Omega$ is positively invariant for model (4). 

Introduce the following notations.
\[
\begin{split}
&\mathbf{I}(t) = (I_1(t), \cdots, I_N(t))^T, \mathbf{P}(t)= (P_1(t), \cdots, P_N(t))^T, \\
&diag\mathbf{I}(t)= diag\left(I_i(t)\right), diag\mathbf{P}(t) = diag\left(P_i(t)\right), \\
&\mathbf{f}(\mathbf{I}(t)) = \left(f_1(\mathbf{I}(t)), \cdots, f_N(\mathbf{I}(t))\right)^T, \\
&\mathbf{g}(\mathbf{P}(t))= \left(g_1(\mathbf{P}(t)), \cdots, g_N(\mathbf{P}(t))\right)^T,\\
&\mathbf{h}(\mathbf{P}(t))= \left(h_1(\mathbf{P}(t)), \cdots, h_N(\mathbf{P}(t))\right)^T.
\end{split}
\]
Then the generic SIPS model can be written as
\begin{equation}
	\left\{
	\begin{aligned}
		\frac{d\mathbf{I}(t)}{dt}&=(\mathbf{E}_N-diag\mathbf{I}(t)-diag\mathbf{P}(t))\mathbf{f}(\mathbf{I}(t))\\
		&\quad- diag\mathbf{I}(t)\mathbf{h}(\mathbf{P}(t))-\mathbf{D}_{\gamma}\mathbf{I}(t),\\
		\frac{d\mathbf{P}(t)}{dt}&=(\mathbf{E}_N-diag\mathbf{I}(t)-diag\mathbf{P}(t))\mathbf{g}(\mathbf{P}(t))\\
		&\quad+diag\mathbf{I}(t)\mathbf{h}(\mathbf{P}(t))-\mathbf{D}_{\alpha}\mathbf{P}(t),
	\end{aligned}
	\right.
\end{equation} 
where $\mathbf{E}_N$ stands for the identity matrix of order $N$. 

\begin{rk}
	When patches are viewed as beneficial viruses, the SIPS models can be regarded as a variant of the bi-virus competing model \cite{Santos2015}. The major difference between them lies in that for the former a node infected with a virus can be 'infected' with the newest patch, whereas for the latter a node infected with one virus cannot be infected with the other one.
\end{rk}

\section{Dynamics of the generic SIPS model}

Consider the generic SIPS model (4). Let $I(t)$ and $P(t)$ denote the fraction at time $t$ of infected nodes and patched nodes in the population, respectively. That is,
\begin{equation*}
	I(t) = \frac{1}{N}\sum_{i=1}^N I_i(t), \quad P(t) = \frac{1}{N}\sum_{i=1}^N P_i(t).
\end{equation*}
The main aim of this work is to determine the development tendency of $I(t)$ and $P(t)$ over time. For that purpose, we need some preliminary knowledge, which are listed below.

\subsection{Preliminaries}

\subsubsection{Matrix theory}

In what follows, we consider only real square matrices. Given a matrix $\mathbf{A}$, let $s(\mathbf{A})$ denote the maximum real part of an eigenvalue of $\mathbf{A}$, and let $\rho(\mathbf{A})$ denote the spectral radius of $\mathbf{A}$, i.e., the maximum modulus of an eigenvalue of $\mathbf{A}$. $\mathbf{A}$ is \emph{Metzler} if its off-diagonal entries are all nonnegative. 

\begin{lm}(Reyleigh Formula, see a corollary of Theorem 4.2.2 in \cite{Horn2013}) 
	Let $\mathbf{A}$ be a real symmetric matrix of order $n$, $\lambda_{\max}(\mathbf{A})$ the maximum eigenvalue of $\mathbf{A}$. Then,
	\begin{equation*} 
		\lambda_{\max}(\mathbf{A}) = \max_{\mathbf{x} \in \mathbb{R}^n, \text{ } \mathbf{x} \neq \mathbf{0}} \frac{\mathbf{x}^T\mathbf{A}\mathbf{x}}{\mathbf{x}^T\mathbf{x}}. 
	\end{equation*}
	Moreover, $\lambda_{\max}(\mathbf{A}) = \frac{\mathbf{x}^T\mathbf{A}\mathbf{x}}{\mathbf{x}^T\mathbf{x}}$ if and only if $\mathbf{x}$ is an eigenvector belonging to $\lambda_{\max}(\mathbf{A})$.
\end{lm}

\begin{lm}(Corollary 8.1.19 in \cite{Horn2013})
	Let $\mathbf{A}$, $\mathbf{B}$ be a pair of nonnegative matrices of the same order. If $\mathbf{A} \leq \mathbf{B}$, then $\rho(\mathbf{A}) \leq \rho(\mathbf{B})$.
\end{lm}

\begin{lm}(Corollary 8.1.30 in \cite{Horn2013})
	Let $\mathbf{A}$ be a nonnegative matrix. If $\mathbf{A}$ has a positive eigenvector $\mathbf{x}$, then (a) $\rho(\mathbf{A})$ is an eigenvalue of $\mathbf{A}$, and (b) $\mathbf{x}$ belongs to $\rho(\mathbf{A})$. 
\end{lm}



\begin{lm}(Lemma 2.3 in \cite{Varga2000}) 
	Let $\mathbf{A}$ be an irreducible Metzler matrix. Then (a) $s(\mathbf{A})$ is a simple eigenvalue of $\mathbf{A}$, and (b) up to scalar multiple, $\mathbf{A}$ has a unique positive eigenvector $\mathbf{x}$ belonging to $s(\mathbf{A})$. 
\end{lm}

A matrix $\mathbf{A}$ is \emph{Hurwitz stable} or simply \emph{Hurwitz} if its eigenvalues all have negative real parts, i.e., $s(\mathbf{A}) < 0$.

\begin{lm}(Chapter 2 in \cite{Horn1991}) 
	A matrix $\mathbf{A}$ is Hurwitz if and only if there is a positive definite matrix $\mathbf{P}$ such that $ \mathbf{A}^T\mathbf{P}+\mathbf{P}\mathbf{A}$ is negative definite.
\end{lm}

A matrix $\mathbf{A}$ is \emph{diagonally stable} if there is a positive definite diagonal matrix $\mathbf{D}$ such that $\mathbf{A}^T\mathbf{D}+\mathbf{D}\mathbf{A}$ is negative definite. Obviously, a diagonally stable matrix is Hurwitz.

\begin{lm}(Section 2 in \cite{Narendra2010}) 
	A Metzler matrix is diagonally stable if it is Hurwitz.
\end{lm}

\begin{lm}(Lemma A.1 in \cite{Khanafer2014a}) 
	Let $\mathbf{A}$ be an irreducible Metzler matrix. If $s(\mathbf{A}) = 0$, then there is a positive definite diagonal matrix $\mathbf{D}$ such that $\mathbf{A}^T\mathbf{D}+\mathbf{D}\mathbf{A}$ is negative semi-definite.
\end{lm}

\subsubsection{Dynamical system theory}

For fundamental theory on differential dynamical systems, see Ref. \cite{Khalil2002}.

\begin{lm} (Chaplygin Lemma, see Theorem 31.4 in \cite{Szarski1965}) Consider a smooth $n$-dimensional system of differential equations 
	\[
	\frac{d\mathbf{x}(t)}{dt} = \mathbf{f}(\mathbf(\mathbf{x}(t)), \quad t \geq 0
	\]
	and the corresponding system of differential inequalities 
	\[
	\frac{d\mathbf{y}(t)}{dt} \leq \mathbf{f}(\mathbf(\mathbf{y}(t)), \quad t \geq 0
	\]
	with $\mathbf{x}(0) = \mathbf{y}(0)$. Suppose
	\[
	\begin{split}
	& f_i(x_1+a_1, \cdots, x_{i-1}+a_{i-1}, x_i, x_{i+1} + a_{i+1}, \cdots, x_n + a_n) \\
	& \geq f_i(x_1, \cdots, x_n), \quad i = 1, \cdots, n, \quad a_1, \cdots, a_n \geq 0.\\
	\end{split}
	\]
	Then $\mathbf{y}(t) \leq \mathbf{x}(t), t \geq 0$.
\end{lm}

\begin{lm} (Strauss-Yorke Theorem, see Corollary 3.3 in \cite{Strauss1967}) Consider a differential dynamical system
	\[
	\frac{d\mathbf{x}(t)}{dt} = \mathbf{f}(\mathbf(\mathbf{x}(t)) + \mathbf{g}(t, \mathbf{x}(t)), \quad t \geq 0,
	\]
	with $\mathbf{g}(t, \mathbf{x}(t)) \rightarrow \mathbf{0}$ when $t \rightarrow \infty$. Let
	\[
	\frac{d\mathbf{y}(t)}{dt} = \mathbf{f}(\mathbf(\mathbf{y}(t)), \quad t \geq 0
	\]
	denote the limit system of this system. If the origin is a global attractor for the limit system, and every solution to the original system is bounded on $[0, \infty)$, then the origin is also a global attractor for the original system.
\end{lm}

\subsubsection{Fixed point theory}


\begin{lm} (Brouwer Fixed Point Theorem, see Theorem 4.10 in \cite{Agarwal2001}) Let $C \subset \mathbb{R}^n$ be nonempty, bounded, closed, and convex. Let $f: C \rightarrow C$ be a continuous function. Then $f$ has a fixed point.
\end{lm}

\subsection{Equilibria}

An equilibrium for a differential dynamical system is a state of the system that is unvaried over time. 
The first step to understanding a differential dynamical system is to examine all of its equilibria. The generic SIPS model might admit four different types of equilibria, which are defined as follows.

\begin{de}
	Let $\mathbf{E} = (I_1, \cdots, I_N, P_1, \cdots, P_N)^T$ be an equilibrium of the generic SIPS model. Let $\mathbf{I} = (I_1, \cdots, I_N)^T$, $\mathbf{P} = (P_1, \cdots, P_N)^T$.
	
	\begin{enumerate}
		
		\item[(a)] $\mathbf{E}$ is \emph{susceptible} if  $\mathbf{I} = \mathbf{P}=\mathbf{0}$, which stands for the steady state of the population for which all nodes are susceptible almost surely.
		
		\item[(b)] $\mathbf{E}$ is \emph{infected} if $\mathbf{I} \neq \mathbf{0}$ and $\mathbf{P} = \mathbf{0}$, which stands for a steady state of the population for which some nodes are infected with positive probability and no node is patched almost surely.
		
		\item[(c)] $\mathbf{E}$ is \emph{patched} if $\mathbf{I} = \mathbf{0}$ and $\mathbf{P} \neq \mathbf{0}$, which stands for a steady state of the population for which some nodes are patched with positive probability and no node is infected almost surely.
		
		\item[(d)] $\mathbf{E}$ is \emph{mixed} if $\mathbf{I} \neq \mathbf{0}$ and $\mathbf{P} \neq \mathbf{0}$, which stands for a steady state of the population for which some nodes are infected with positive probability and some nodes are patched with positive probability.
		
	\end{enumerate}
	
\end{de}

Obviously, the generic SIPS model always admits the susceptible equilibrium, denoted $\mathbf{E}_s=(0,\cdots,0)^T$. Due to the complexity of the model, we are unable to figure out its equilibria other than the susceptible one. For the purpose of examining all possible equilibria of the model, we need to define a pair of Metzler matrices as follows.
\begin{equation}
	\mathbf{Q}_1=\frac{\partial \mathbf{\mathbf{f}(\mathbf{0})}}{\partial \mathbf{x}}-\mathbf{D}_{\gamma},\quad \mathbf{Q}_2=\frac{\partial{\mathbf{g}(\mathbf{0})}}{\partial\mathbf{x}}-\mathbf{D}_{\alpha},
\end{equation}
where $\frac{\partial \mathbf{\mathbf{f}(\mathbf{0})}}{\partial \mathbf{x}} = \left(\frac{\partial f_i(\mathbf{0})}{\partial x_j}\right)_{N \times N}$ and $\frac{\partial {\mathbf{g}(\mathbf{0})}}{\partial \mathbf{x}} = \left(\frac{\partial g_i(\mathbf{0})}{\partial x_j}\right)_{N \times N}$ stand for the Jacobian matrix of $\mathbf{f}$ and $\mathbf{g}$ evaluated at the origin, respectively. As $G_v$ is strongly connected, $\mathbf{Q}_1$ is irreducible. As $G_p$ is strongly connected, $\mathbf{Q}_2$ is irreducible.

For the linear SIPS model, we have

\[
\mathbf{Q}_1=\mathbf{M}_{\beta}-\mathbf{D}_{\gamma},\quad \mathbf{Q}_2=\mathbf{M}_{\delta_1}-\mathbf{D}_{\alpha}.
\]

The following theorem is a fundamental result about the equilibria of the generic SIPS model.

\begin{thm}
	Consider model (4). The following claims hold.
	\begin{enumerate}
		\item [(a)] If
		$s(\mathbf{Q}_1)>0$, then there is a unique infected equilibrium, denoted $\mathbf{E}_i=(I_1^*,\cdots,I_N^*,0,\cdots,0)^T$. Moreover, $0 < I_i^* < 1$, $1 \leq i \leq N$. Let $\mathbf{I}^* = (I_1^*,\cdots,I_N^*)^T$.
		\item [(b)] If
		$s(\mathbf{Q}_2)>0$, then there is a unique patched equilibrium, denoted $\mathbf{E}_p=(0,\cdots,0,P_1^*,\cdots,P_N^*)^T$. Moreover, $0 < P_i^* < 1$, $1 \leq i \leq N$. Let $\mathbf{P}^* = (P_1^*,\cdots,P_N^*)^T$.
		\item [(c)] Suppose $\mathbf{g}=\mathbf{h}$. If $s(\mathbf{Q}_2)>0$ and
		\[ s(\mathbf{Q}_1-(\mathbf{Q}_1+\mathbf{D}_{\gamma}) diag\mathbf{P}^* -diag\mathbf{g}(\mathbf{P}^*))>0, 
		\]
		then there is a unique mixed equilibrium, denoted $\mathbf{E}_m=(I_1^{**},\cdots,I_N^{**},P_1^{**},\cdots,P_N^{**})^T$. Moreover, $P_i^{**} = P_i^{*}$, $0 < I_i^{**} < 1 - P_i^{*}$, $1\leq i\leq N$.
	\end{enumerate}
\end{thm}

\begin{IEEEproof}
	(a) Suppose model (4) admits an infected equilibrium
	$\mathbf{E}=(I_1,\cdots,I_N,0,\cdots, 0)^T$. We show that $0 < I_i < 1$, $1\leq i \leq N$. It follows from the model that 
	\[
	I_i=\frac{f_i(I_1,\cdots,I_N)}{\gamma_i+f_i(I_1,\cdots,I_N)}< 1, \quad 1 \leq i \leq N.
	\]
	On the contrary, suppose that some $I_k = 0$. It follows from model (4) that $f_k(I_1,\cdots,I_N)=0$. As $G_v$ is strongly connected, we get that some $\beta_{kl} >0$, implying that $I_l = 0$. Repeating this argument, we finally get that $I_i = 0$, $1\leq i \leq N$, contradicting the assumption that  $\mathbf{E}$ is an infected equilibrium. Hence, $I_i > 0$, $1\leq i \leq N$.
	
	Define a continuous mapping $\mathbf{T}=(T_1,\cdots,T_N)^T: (0,1]^N \rightarrow (0,1]^N$ by 
	\[
	T_i(\mathbf{x})=\frac{f_i(\mathbf{x})}{\gamma_i+f_i(\mathbf{x})}, \quad \mathbf{x}=(x_1,\cdots,x_N)^T\in(0,1]^N
	\].
	It suffices to show that $\mathbf{T}$ has a unique fixed point.
	Let $\mathbf{P}(t) \equiv \mathbf{0}$ and rewrite model (4) as
	$\frac{d\mathbf{I}(t)}{dt}=\mathbf{Q}_1\mathbf{I}(t)+\mathbf{G}(\mathbf{I}(t))$,
	where $\mathbf{G}(\mathbf{I}(t))=o(\|\mathbf{I}(t)\|)$. By Lemma 5, $\mathbf{Q}_1$ has a positive eigenvector $\mathbf{v}=(v_1,\cdots,v_N)^T$ belonging to eigenvalue $s(\mathbf{Q}_1)$. As $s(Q_1) > 0$, we have $\mathbf{Q}_1\mathbf{v}=s(\mathbf{Q}_1)\mathbf{v} > \mathbf{0}$. Hence, there is a small $\varepsilon > 0$ such that
	$\mathbf{Q}_1 \cdot (\varepsilon \mathbf{v})+\mathbf{G}(\varepsilon \mathbf{v})=\varepsilon s(\mathbf{Q}_1)\mathbf{v}+\mathbf{G}(\varepsilon \mathbf{v})\geq \mathbf{0}$,
	which is equivalent to $\mathbf{T}(\varepsilon \mathbf{v})\geq \varepsilon \mathbf{v}$. On the other hand, it is easily verified that $\mathbf{T}$ is monotonically increasing, i.e., $\mathbf{u} \geq \mathbf{w}$ implies $\mathbf{T}(\mathbf{u}) \geq \mathbf{T}(\mathbf{w})$.  Define a compact convex set as $
	K=\prod_{i=1}^N[\varepsilon v_i,1]$. Then $\mathbf{T}|_{K}$
	maps $K$ into $K$. It follows from Lemma 11 that $\mathbf{T}$ has a fixed point $\mathbf{I}^*=(I_1^*,\cdots,I_N^*)^T$ in $K$. 
	
	Suppose $\mathbf{T}$ has a fixed point $\mathbf{I}^{**}=(I_1^{**},\cdots,I_N^{**})^T$ other than $\mathbf{I}^{*}$. Let $\theta=\max_{1 \leq i \leq N}\frac{I_i^*}{I_i^{**}}$, $i_0=\arg\max_{1 \leq i \leq N}\frac{I_i^*}{I_i^{**}}$. Without loss of generality, assume $\theta>1$. It follows that 
	\[
	\begin{split}
	I_{i_0}^{*}=& T_{i_0}(\mathbf{I}^*)\leq T_{i_0}(\theta\mathbf{I}^{**}) \\
	=&\frac{f_{i_0}(\theta\mathbf{I}^{**})}{\gamma_{i_0}+f_{i_0}(\theta\mathbf{I}^{**})}< \frac{f_{i_0}(\theta \mathbf{I}^{**})}{\gamma_{i_0}+f_{i_0}(\mathbf{I}^{**})} \\
	\leq &\frac{\theta f_{i_0}(\mathbf{I}^{**})}{\gamma_{i_0}+f_{i_0}(\mathbf{I}^{**})}=\theta T_{i_0}(\mathbf{I}^{**})=\theta I_{i_0}^{**},
	\end{split}
	\]
	where $f_{i_0}(\theta\mathbf{I}^{**}) \leq \theta f_{i_0}(\mathbf{I}^{**})$ follows from the concavity of $f_{i_0}$. This contradicts the assumption that $I_{i_0}^{*}=\theta I_{i_0}^{**}$. Hence, $\mathbf{I}^{*}$ is the unique fixed point of $\mathbf{T}$. The proof is complete.

	(b) The argument is analogous to that for Claim (a) and hence is
	omitted.
	
	(c) Suppose model (4) admits a mixed equilibrium
	$\mathbf{E}=(I_1,\cdots,I_N,P_1,\cdots, P_N)^T$. By an argument analogous to that for Claim (b), we get that $P_i = P_i^*$, $1 \leq i \leq N$. It follows from the first $N$ equations of model (4) that
	\[
	I_i=\frac{ (1-P_i^*)f_i(I_1,\cdots,I_N)}{\gamma_i+f_i(I_1,\cdots,I_N)+h_i(\mathbf{P}^*)},
	\quad\quad 1\leq i\leq N.
	\]
	The subsequent argument is analogous to that for Claim (a).
\end{IEEEproof}

This theorem manifests that the existence and locations of equilibria of the generic SIPS model are dependent in a complex way upon the model parameters. 

For later use, let
\[
\begin{split}
I^* &= \frac{1}{N}\sum_{i=1}^N I_i^*, \quad P^* = \frac{1}{N}\sum_{i=1}^N P_i^*, \\
I^{**} &= \frac{1}{N}\sum_{i=1}^N I_i^{**}, \quad P^{**} = \frac{1}{N}\sum_{i=1}^N P_i^{**}.
\end{split}
\]

\begin{rk}
	As a decentralized patch distribution protocol is commonly designed so that multiple copies of a patch are forwarded from a patched node to some or all of its unpatched neighbors without further differentiating whether such a neighbor is susceptible or infected, in many real scenarios the condition of $\mathbf{g}=\mathbf{h}$ in claim (c) of Theorem 1 can be regarded as being true.
\end{rk}

\subsection{Attractivity analysis}

If an equilibrium of a differential dynamical system is a global attractor, then every orbit of the system would approach the equilibrium. 
Now, let us examine the global attractivity of the equilibria of the generic SIPS model. First, we have the following criterion for the global attractivity of the susceptible equilibrium of the generic SIPS model.

\begin{thm}
	Consider model (4). Suppose $s(\mathbf{Q}_1)\leq 0$ and $s(\mathbf{Q}_2)\leq 0$. Then the susceptible equilibrium $\mathbf{E}_s$ is the attractor for $\Omega$. Hence, $I(t) \rightarrow 0$ and $P(t) \rightarrow 0$ as $t \rightarrow \infty$.
\end{thm}

\begin{IEEEproof}
	Let $(\mathbf{I}(t)^T, \mathbf{P}(t)^T)^T$ be a solution to model (4). It follows from the first $N$ equations of model (4) that
	$\frac{d\mathbf{I}(t)}{dt}\leq (\mathbf{E}_N-diag\mathbf{I}(t)) \mathbf{f} (\mathbf{I}(t))-\mathbf{D}_{\gamma} \mathbf{I}(t)$. Consider the comparison system
	\begin{equation}
		\frac{d\mathbf{u}(t)}{dt}=(\mathbf{E}_N-diag\mathbf{u}(t))\mathbf{f}(\mathbf{u}(t))-\mathbf{D}_{\gamma} \mathbf{u}(t)
	\end{equation}
	with $\mathbf{u}(0) = \mathbf{I}(0)$. This system admits the trivial equilibrium $\mathbf{0}$. By Lemma 9, we have $\mathbf{u}(t) \geq \mathbf{I}(t) \geq 0, t \geq 0$. We proceed by distinguishing two possibilities. 
	
	\emph{Case 1}: $s(\mathbf{Q}_1)<0$. By Lemma 7, there is a positive definite diagonal matrix $\mathbf{P}_1$ such that $\mathbf{Q}_1^T\mathbf{P}_1+\mathbf{P}_1\mathbf{Q}_1$ is negative definite. Let $\mathbf{u}=(u_1,\cdots,u_N)^T$, and define a
	positive definite function as
	$V_1(\mathbf{u})=\mathbf{u}^T \mathbf{P}_1 \mathbf{u}$. By calculations, we get
	\[
	\begin{split}
	\frac{dV_1(\mathbf{u}(t))}{dt}\mid_{(7)} &= 2\mathbf{u}(t)^T \mathbf{P}_1 \frac{d\mathbf{u}(t)}{dt}\\
	&\leq 2\mathbf{u}(t)^T \mathbf{P}_1\left[\mathbf{f}(\mathbf{u}(t)) - \mathbf{D}_{\gamma}\mathbf{u}(t) \right] \\
	&\leq 2\mathbf{u}(t)^T \mathbf{P}_1\mathbf{Q}_1 \mathbf{u}(t)\\
	&= \mathbf{u}(t)^T[\mathbf{Q}_1^T\mathbf{P}_1+\mathbf{P}_1\mathbf{Q}_1]\mathbf{u}(t) \leq 0.
	\end{split}
	\]
	
	\noindent where the second inequality follows from the concavity of $\mathbf{f}(\mathbf{x}) - \mathbf{D}_{\gamma}\mathbf{x}$. Furthermore, 
	$\frac{dV_1(\mathbf{u}(t))}{dt}\mid_{(7)}=0$ if and
	only if $\mathbf{u}(t)=\mathbf{0}$. According to the
	LaSalle Invariance Principle (Corollary 4.1 in \cite{Khalil2002}), the trivial equilibrium $\mathbf{0}$ of system (8) is asymptotically stable for $[0,1]^N$.
	
	\emph{Case 2}: $s(\mathbf{Q}_1)=0$. By Lemma 8, there is a positive definite diagonal matrix $\mathbf{P}_2$ such that $\mathbf{Q}_1^T\mathbf{P}_2+\mathbf{P}_2\mathbf{Q}_1$ is negative semi-definite. Define a
	positive definite function as $
	V_2(\mathbf{u})=\mathbf{u}^T \mathbf{P}_2 \mathbf{u}$. Similarly, we have
	\[
	\begin{split}
	\frac{dV_2(\mathbf{u}(t))}{dt}\mid_{(7)}
	\leq& \mathbf{u}(t)^T[\mathbf{Q}_1^T\mathbf{P}_2+\mathbf{P}_2\mathbf{Q}_1]\mathbf{u}(t) \\
	& -2\mathbf{u}(t)^T \mathbf{P}_2\cdot diag\mathbf{u}(t)\cdot \mathbf{f}(\mathbf{u}(t)) \leq 0.
	\end{split}
	\]

	If $\mathbf{Q}_1^T\mathbf{P}_2+\mathbf{P}_2\mathbf{Q}_1$ is negative definite, the subsequent argument is analogous to that for Case 1. Now, assume $\mathbf{Q}_1^T\mathbf{P}_2+\mathbf{P}_2\mathbf{Q}_1$ is not negative definite, which implies $s(\mathbf{Q}_1^T\mathbf{P}_2+\mathbf{P}_2\mathbf{Q}_1)=0$. As $\mathbf{Q}_1^T\mathbf{P}_2+\mathbf{P}_2\mathbf{Q}_1$ is Metzler and irreducible, it follows from Lemma 5 that (a) $0$ is a simple eigenvalue of $\mathbf{Q}_1^T\mathbf{P}_2+\mathbf{P}_2\mathbf{Q}_1$, and (b) up to scalar multiple, $\mathbf{Q}_1^T\mathbf{P}_2+\mathbf{P}_2\mathbf{Q}_1$ has a positive eigenvector belonging to eigenvalue 0. 
	
	Obviously, 
	$\frac{dV_2(\mathbf{u}(t))}{dt}\mid_{(7)}=0$ if $\mathbf{u}(t)=\mathbf{0}$. On the contrary, suppose $\frac{dV_2(\mathbf{u}(t))}{dt}\mid_{(7)}=0$ for some $\mathbf{u}(t)\geq \mathbf{0}$. If $\mathbf{u}(t) > \mathbf{0}$, then $\mathbf{f}(\mathbf{u}(t))>\mathbf{0}$, implying $\frac{dV_2(\mathbf{u}(t))}{dt}\mid_{(7)}<0$, a contradiction. If $\mathbf{u}(t)$ has a zero component, then $\mathbf{u}(t)$ is not an eigenvector of $\mathbf{Q}_1^T\mathbf{P}_2+\mathbf{P}_2\mathbf{Q}_1$ belonging to eigenvalue 0. It follows from Lemma 2 that  $\mathbf{u}(t)^T[\mathbf{Q}_1^T\mathbf{P}_2+\mathbf{P}_2\mathbf{Q}_1]\mathbf{u}(t)<0$, implying $\frac{dV_2(\mathbf{u}(t))}{dt}\mid_{(7)}<0$, again a contradiction. Hence, $\mathbf{u}(t)=\mathbf{0}$ if $\frac{dV_2(\mathbf{u}(t))}{dt}\mid_{(7)}=0$. It follows from the
	LaSalle Invariance Principle that the trivial equilibrium $\mathbf{0}$ of system (7) is asymptotically stable with respect to $[0,1]^N$.
	
	Combining Cases 1 and 2, we get $\mathbf{u}(t) \rightarrow \mathbf{0}$. According to Lemma 9, we get $\mathbf{I}(t)\leq \mathbf{u}(t)$, which implies $\mathbf{I}(t) \rightarrow \mathbf{0}$. Consider the following limit system of model (4).
	\begin{equation}
		\frac{d\mathbf{w}(t)}{dt}=(\mathbf{E}_N-diag\mathbf{w}(t))\mathbf{g}(\mathbf{w}(t))-\mathbf{D}_{\alpha} \mathbf{w}(t)
	\end{equation}
	
	\noindent with $\mathbf{w}(0) = \mathbf{P}(0)$. As with the comparison system (7), we have $\mathbf{w}(t) \rightarrow \mathbf{0}$ when $t \rightarrow \infty$. It follows from Lemma 10 that $\mathbf{P}(t)\rightarrow \mathbf{0}$ when $t \rightarrow \infty$. The proof is complete.
\end{IEEEproof}

This theorem has the following useful corollaries.

\begin{cor}
	The susceptible equilibrium $\mathbf{E}_s$ of model (4) is the attractor for $\Omega$ if one of the following conditions is satisfied.
	
	\begin{enumerate}
		
		\item[(a)] $\rho(\mathbf{Q}_1\mathbf{D}_{\gamma}^{-1}+\mathbf{E}_N) < 1$, $\rho(\mathbf{Q}_2\mathbf{D}_{\alpha}^{-1}+\mathbf{E}_N) < 1$.
		
		\item[(b)] $\rho(\mathbf{M}_{\beta}\mathbf{D}_{\gamma}^{-1}) < 1$, $\rho(\mathbf{M}_{\delta_1}\mathbf{D}_{\alpha}^{-1}) < 1$.
		
		\item[(c)] $\sum_{i=1}^{N}\beta_{ij} < \gamma_j$, $\sum_{i=1}^{N}\delta_{ij}^{(1)} < \alpha_j$, $j = 1, 2 \cdots, N$.
		
		\item[(d)] $\sum_{j=1}^{N}\frac{\beta_{ij}}{\gamma_j} < 1$, $\sum_{j=1}^{N}\frac{\delta_{ij}^{(1)}}{\alpha_j} < 1$, $i = 1, 2, \cdots, N$.
		
	\end{enumerate}
\end{cor}

\begin{IEEEproof}
	(a) We first show $s(\mathbf{Q}_1) < 0$. As $\mathbf{Q}_1\mathbf{D}_{\gamma}^{-1}$ is Metzler and irreducible, it follows from Lemma 5 that $\mathbf{Q}_1\mathbf{D}_{\gamma}^{-1}$ has a unique positive eigenvector $\mathbf{x}$ belonging to eigenvalue $s(\mathbf{Q}_1\mathbf{D}_{\gamma}^{-1})$. So, 
	$(\mathbf{Q}_1\mathbf{D}_{\gamma}^{-1}+\mathbf{E}_N)\mathbf{x}=[s(\mathbf{Q}_1\mathbf{D}_{\gamma}^{-1})+1]\mathbf{x}$. 
	That is, $\mathbf{x}$ is an eigenvector of $\mathbf{Q}_1\mathbf{D}_{\gamma}^{-1}+\mathbf{E}_N$ belonging to eigenvalue $s(\mathbf{Q}_1\mathbf{D}_{\gamma}^{-1})+1$. It follows from Lemma 4 that 
	$s(\mathbf{Q}_1\mathbf{D}_{\gamma}^{-1}) = \rho(\mathbf{Q}_1\mathbf{D}_{\gamma}^{-1}+\mathbf{E}_N)-1 < 0$. 
	By Lemma 7, there is a positive definite diagonal matrix $\mathbf{D}$ such that $
	\mathbf{P}=(\mathbf{Q}_1\mathbf{D}_{\gamma}^{-1})^T\mathbf{D}+\mathbf{D}(\mathbf{Q}_1\mathbf{D}_{\gamma}^{-1})$ is negative definite. Direct calculations give 
	$
	\left[\mathbf{D}_{\gamma}^{\frac{1}{2}}\mathbf{Q}_1\mathbf{D}_{\gamma}^{-\frac{1}{2}}\right]^T\mathbf{D}+\mathbf{D}\left[\mathbf{D}_{\gamma}^{\frac{1}{2}}\mathbf{Q}_1\mathbf{D}_{\gamma}^{-\frac{1}{2}}\right] 
	=\mathbf{D}_{\gamma}^{\frac{1}{2}}\mathbf{P}\mathbf{D}_{\gamma}^{\frac{1}{2}}$. As $\mathbf{D}_{\gamma}^{\frac{1}{2}}\mathbf{P}\mathbf{D}_{\gamma}^{\frac{1}{2}}$ is negative definite, $\mathbf{D}_{\gamma}^{\frac{1}{2}}\mathbf{Q}_1\mathbf{D}_{\gamma}^{-\frac{1}{2}}$ is diagonally stable and hence Hurwitz. It follows that $s(\mathbf{Q}_1) = s(\mathbf{D}_{\gamma}^{\frac{1}{2}}\mathbf{Q}_1\mathbf{D}_{\gamma}^{-\frac{1}{2}}) < 0$.
	
	Similarly, we have $s(\mathbf{Q}_2) < 0$. The declared result follows from Theorem 2.
	
	(b) By the concavity of $f_i(\mathbf{x})$, we have
	$\frac{\partial f_i(\mathbf{0})}{\partial x_j}\leq \beta_{ij}$. That is, $\mathbf{Q}_1 + \mathbf{D}_{\gamma} \leq \mathbf{M}_{\beta}$. Hence, 
	$\rho(\mathbf{Q}_1\mathbf{D}_{\gamma}^{-1}+\mathbf{E}_N) \leq \rho(\mathbf{M}_{\beta}\mathbf{D}_{\gamma}^{-1}) < 1$.
	Similarly, we have $\rho(\mathbf{Q}_2\mathbf{D}_{\alpha}^{-1}+\mathbf{E}_N) < 1$. The claim follows from Claim (a) of this corollary.
	
	(c) The claim follows from Claim (b) of this corollary and $\rho(\mathbf{M}) \leq ||\mathbf{M}||_1$.
	
	(d) The claim follows from Claim (b) of this corollary and $\rho(\mathbf{M}) \leq ||\mathbf{M}||_{\infty}$.
\end{IEEEproof}

Next, define a function as
\[
\mathbf{g}^*(\mathbf{x})=\left(\max\{g_1(\mathbf{x}),h_1(\mathbf{x})\},\cdots,\max\{g_N(\mathbf{x}),h_N(\mathbf{x})\}\right)^T,
\]
and define a pair of Metzler matrices as
\begin{equation}
	\mathbf{Q}_3=\frac{\partial \mathbf{\mathbf{h}(\mathbf{0})}}{\partial \mathbf{x}}-\mathbf{D}_{\alpha}, \quad \mathbf{Q}_4=\frac{\partial \mathbf{\mathbf{g}^*(\mathbf{0})}}{\partial \mathbf{x}}-\mathbf{D}_{\alpha}
\end{equation}
where $\frac{\partial \mathbf{\mathbf{h}(\mathbf{0})}}{\partial \mathbf{x}} = \left(\frac{\partial h_i(\mathbf{0})}{\partial x_j}\right)_{N \times N}$ and $\frac{\partial \mathbf{\mathbf{g}^*(\mathbf{0})}}{\partial \mathbf{x}} = \left(\frac{\partial g_i^*(\mathbf{0})}{\partial x_j}\right)_{N \times N}$ stand for the Jacobian matrix of $\mathbf{h}$ and $\mathbf{g}^*$ evaluated at the origin, respectively. As $G_p$ is strongly connected, $\mathbf{Q}_3$ and $\mathbf{Q}_4$ are irreducible.

For the linear SIPS model, we have

\[
\mathbf{Q}_3=\mathbf{M}_{\delta_2}-\mathbf{D}_{\alpha}, \quad \mathbf{Q}_4=\left(\max\left\{\delta^{(1)}_{ij}, \delta^{(2)}_{ij}\right\}\right)_{N \times N}-\mathbf{D}_{\alpha}.
\]

The following theorem offers a criterion for the global attractivity of the infected equilibrium of the generic SIPS model, if any.

\begin{thm}
	Consider model (4). Suppose $s(\mathbf{Q}_1)> 0$ and $s(\mathbf{Q}_4)\leq 0$. Then the infected equilibrium $\mathbf{E}_i$ is an attractor for $\{(\mathbf{I},\mathbf{P})\in\Omega:\mathbf{I}\neq \mathbf{0}\}$. Hence, if $I(0) \neq 0$, then $I(t) \rightarrow I^*$ and $P(t) \rightarrow 0$ as $t \rightarrow \infty$.
\end{thm}

\begin{IEEEproof}
	Let $(\mathbf{I}(t)^T, \mathbf{P}(t)^T)^T$ be a solution
	to model (4) with $\mathbf{I}(0)\neq \mathbf{0}$. It follows from the last $N$ equations of model (4) that 
	$\frac{d\mathbf{P}(t)}{dt}\leq(\mathbf{E}_N-diag(\mathbf{P}(t)))\mathbf{g}^*(\mathbf{P}(t))-\mathbf{D}_{\alpha}\mathbf{P}(t)$.
	Similar to the argument for Theorem 2, we get $\mathbf{P}(t)\rightarrow\mathbf{0}$. Consider the limit system (7)
	with $\mathbf{u}(0) = \mathbf{I}(0)$. Theorem 1 confirms that the system admits a unique nonzero equilibrium $\mathbf{I}^*=(I_1^*,\cdots,I_N^*)$. By Lemma 10, it suffices to show that $\mathbf{I}^*$ is asymptotically stable for $(0,1]^N$. Given a solution $\mathbf{u}(t) = (u_1(t), \cdots, u_N(t))^T$ to system (7) with $\mathbf{u}(0) > \mathbf{0}$. First, let us show the following claim.
	
	\emph{Claim 1:} $\mathbf{u}(t) > \mathbf{0}$ for $t > 0$.
	
	\emph{Proof of Claim 1:} On the contrary, suppose there is $t_0 > 0$ such that (a) $\mathbf{u}(t) > \mathbf{0}$, $0 < t < t_0$, and (b) $u_i(t_0)=0$ for some $i$. According to the smoothness of $\mathbf{u}(t)$, we get $\frac{du_{i}(t_0)}{dt}=0$, implying $f_{i}(\mathbf{u}(t_0))=0$. As $G_v$ is strongly connected, there is $j$ such that $\beta_{ij}>0$, which implies $u_{j}(t_0)=0$. Working inductively, we conclude that $\mathbf{u}(t_0)=0$. This contradicts the uniqueness of solution to system (7) with given initial conditions. Claim 1 is proven.
	
	For $t > 0$, let $
	Z(\mathbf{u}(t))=\max_{1 \leq k \leq N} \frac{u_k(t)}{I_k^*}$, $
	z(\mathbf{u}(t))=\min_{1 \leq k \leq N} \frac{u_k(t)}{I_k^*}$.
	Define a function $V_3$ as
	\[
	V_3(\mathbf{u}(t))=\max\{Z(\mathbf{u}(t))-1,0\}+\max\{1-z(\mathbf{u}(t)),0\}.
	\]
	
	\noindent It is easily verified that $V_3$ is positive definite with respect to $\mathbf{I}^*$, i.e., (a) $V_3(\mathbf{u}(t))\geq 0$, and (b) $V_3(\mathbf{u}(t))=0$ if and only if $\mathbf{u}(t)=\mathbf{I}^{*}$. Next , let us show that $D^+V_3(\mathbf{u}(t)) \leq 0$, where $D^+$ stands for the upper right Dini derivative. To this end, we need to show the following two claims for $t > 0$.
	
	\emph{Claim 2:} $D^+Z(\mathbf{u}(t))\leq0$ if $Z(\mathbf{u}(t))\geq1$.
	Moreover, $D^+Z(\mathbf{u}(t))<0$ if $Z(\mathbf{u}(t))>1$.
	
	\emph{Claim 3:} $D_+z(\mathbf{u}(t))\geq0$ if $z(\mathbf{u}(t))\leq1$.
	Moreover, $D_+z(\mathbf{u}(t))> 0$ if  $z(\mathbf{u}(t))<1$. Here $D_+$ stands for the lower right Dini derivative.
	
	\emph{Proof of Claim 2:} Choose $k_0$ such that  
	$
	Z(\mathbf{u}(t))=\frac{u_{k_0}(t)}{I_{k_0}^*}$, $D^+Z(\mathbf{u}(t))=\frac{u_{k_0}^{'}(t)}{I_{k_0}^*}$. Then, 
	$\frac{I_{k_0}^{*}}{u_{k_0}(t)}u_{k_0}^{'}(t)
	=\left(1-u_{k_0}(t)\right)\frac{I_{k_0}^{*}}{u_{k_0}(t)}f_{k_0}(\mathbf{u}(t))-\gamma_{k_0}I_{k_0}^{*}$. If $f_{k_0}(\mathbf{u}(t))=0$, then
	$\frac{I_{k_0}^{*}}{u_{k_0}(t)}u_{k_0}^{'}(t) <0$, which implies $D^+Z(\mathbf{u}(t))<0$. Now assume $f_{k_0}(\mathbf{u}(t))>0$, then
	\[
	\begin{split}
	\frac{I_{k_0}^{*}}{u_{k_0}(t)}u_{k_0}^{'}(t)\leq&
	(1-I_{k_0}^{*})\frac{I_{k_0}^{*}}{u_{k_0}(t)}f_{k_0}(\mathbf{u}(t))-\gamma_{k_0}I_{k_0}^{*}\\
	\leq &(1-I_{k_0}^{*})f_{k_0}\left(\frac{I_{k_0}^{*}}{u_{k_0}(t)}\mathbf{u}(t)\right)-\gamma_{k_0}I_{k_0}^{*}\\
	\leq&(1-I_{k_0}^{*})f_{k_0}\left(\mathbf{I}^*\right)-\gamma_{k_0}I_{k_0}^{*}=0,
	\end{split}
	\]
	where the second inequality follows from the concavity of $f_{k_0}$, and the third inequality follows from the monotonicity of $f_{k_0}$. This implies $D^+Z(\mathbf{u}(t))\leq0$. Noting that the first inequality is strict if $Z(\mathbf{u}(t))>1$, we get that $D^+Z(\mathbf{u}(t))<0$ if $Z(\mathbf{u}(t))>1$. Claim 2 is proven.
	
	The argument for Claim 3 is analogous to that for Claim 2 and hence is omitted. Next, consider three possibilities.
	
	Case 1: $Z(\mathbf{u}(t)) < 1$. Then $z(\mathbf{u}(t)) < 1$ and $V_3(\mathbf{u}(t)) = 1 - z(\mathbf{u}(t))$. Hence, 
	$D^+V_3(\mathbf{u}(t)) = -D_+z(\mathbf{u}(t)) < 0$.
	
	Case 2: $z(\mathbf{u}(t)) > 1$. Then $Z(\mathbf{u}(t)) > 1$ and $V_3(\mathbf{u}(t)) = Z(\mathbf{u}(t)) - 1$. Hence, 
	$D^+V_3(\mathbf{u}(t)) = D^+Z(\mathbf{u}(t)) < 0$.
	
	Case 3 If $Z(\mathbf{u}(t)) \geq 1$, $z(\mathbf{u}(t)) \leq 1$. Then $V_3(\mathbf{u}(t)) = Z(\mathbf{u}(t)) - z(\mathbf{u}(t))$ and 
	$D^+V_3(\mathbf{u}(t)) = D^+Z(\mathbf{u}(t)) - D_+z(\mathbf{u}(t)) \leq 0$.
	Moreover, the equality holds if and only if $\mathbf{u}(t) = \mathbf{I}^*$.
	
	The declared result follows from the LaSalle Invariance Principle.
\end{IEEEproof}
This theorem has the following direct corollary.

\begin{cor}
	The infected equilibrium $\mathbf{E}_i$ of model (4) is the attractor for $\{(\mathbf{I},\mathbf{P})\in\Omega:\mathbf{I}\neq \mathbf{0}\}$ if one of the following two sets of conditions is satisfied.
	\begin{enumerate}
		
		\item[(a)] $\mathbf{g}\geq \mathbf{h}$, $s(\mathbf{Q}_1)> 0$, and $s(\mathbf{Q}_2)\leq 0$.
		
		\item[(b)] $\mathbf{g}\leq \mathbf{h}$, $s(\mathbf{Q}_1)> 0$, and $s(\mathbf{Q}_3)\leq 0$.
		
	\end{enumerate}
\end{cor}

Next, we have the following criterion for the global attractivity of the patched equilibrium of the generic SIPS model. 

\begin{thm}
	Consider model (4). Suppose $s(\mathbf{Q}_1)\leq 0$ and $s(\mathbf{Q}_2)> 0$. Then the patched equilibrium $\mathbf{E}_p$ is the attractor for $\{(\mathbf{I},\mathbf{P})\in\Omega:\mathbf{P}\neq \mathbf{0}\}$. Hence, if $P(0) \neq 0$, then $I(t) \rightarrow 0$ and $P(t) \rightarrow P^*$ as $t \rightarrow \infty$.
\end{thm}

The argument for the theorem is analogous to that for Theorem 3 and hence is omitted. 

Finally, the following theorem gives a criterion for the global attractivity of the mixed equilibrium of the generic SIPS model.

\begin{thm}
	Consider model (4) with $\mathbf{g}=\mathbf{h}$. Suppose $s(\mathbf{Q}_2)>0$ and $
	s\left(\mathbf{Q}_1- (\mathbf{Q}_1+\mathbf{D}_{\gamma})diag\mathbf{P}^*-diag\mathbf{g}(\mathbf{P}^*)\right)>0$. 
	Then the mixed equilibrium $\mathbf{E}_m$ is the attractor for $\{(\mathbf{I},\mathbf{P})\in\Omega:\mathbf{I}\neq \mathbf{0},\mathbf{P}\neq \mathbf{0}\}$. Hence, if $I(0) \neq 0$ and $P(0) \neq 0$, then $I(t) \rightarrow I^{**}$ and $P(t) \rightarrow P^{**}$ as $t \rightarrow \infty$.
\end{thm}

\begin{IEEEproof}
	Let $(\mathbf{I}(t)^T, \mathbf{P}(t)^T)^T$ be a solution
	to model (4) with $\mathbf{I}(0)\neq\mathbf{0}$ and $\mathbf{P}(0)\neq\mathbf{0}$. As $\mathbf{g}=\mathbf{h}$, the last $N$ equations of model (4) reduce to 
	$\frac{d\mathbf{P}(t)}{dt}=(\mathbf{E}_N-diag(\mathbf{P}(t))g(\mathbf{P}(t))-\mathbf{D}_{\alpha}\mathbf{P}(t)$.
	\noindent Similar to the argument for Theorem 4, we get $\mathbf{P}(t)\rightarrow\mathbf{P}^*$. Consider the following limit system of model (4). 
	\[
	\begin{split}
	\frac{d\mathbf{z}(t)}{dt}=&(\mathbf{E}_N-diag(\mathbf{P}^*)-diag(\mathbf{z}(t)))\mathbf{f}(\mathbf{z}(t))\\ \nonumber
	&-(diag\mathbf{h}(\mathbf{P}^*)+\mathbf{D}_{\gamma})\mathbf{z}(t), \mathbf{z}(0) = \mathbf{I}(0).
	\end{split}
	\]
	By Theorem 1, $\mathbf{I}^{**}$ is an equilibrium of the system. The subsequent treatment is analogous to that of system (7). 
\end{IEEEproof}

The ultimate goal to achieve by studying the dynamics of the generic SIPS model is to help us work out effective strategies of wiping out the digital viruses in the population under consideration. By combining Theorem 2 with Theorem 4, we find that the viruses in the population would die out if $s(\mathbf{Q}_1) \leq 0$. Hence, the key to the virus suppression is to reduce $s(\mathbf{Q}_1)$. The next theorem reveals the relationship between $s(\mathbf{Q}_1)$ and the spectral radius of another matrix.

\begin{thm}
	Consider model (4) and let $\overline{\gamma} = \max_{1 \leq i \leq N}\gamma_i$. Then 
	$
	s(\mathbf{Q}_1)=\rho(\mathbf{Q}_1+\overline{\gamma}\mathbf{E}_N)-\overline{\gamma}.
	$
\end{thm}

\begin{IEEEproof}
	It follows from Lemma 5 that $\mathbf{Q}_1$ has a positive eigenvector $\mathbf{x}$ belonging to the eigenvalue $s(\mathbf{Q}_1)$. Thus,
	$(\mathbf{Q}_1+\overline{\gamma} \mathbf{E}_N)\mathbf{x}=[s(\mathbf{Q}_1)+\overline{\gamma}]\mathbf{x}$. 
	That is, $\mathbf{x}$ is an eigenvector of $\mathbf{Q}_1+\overline{\gamma}\mathbf{E}_N$ belonging to the eigenvalue $s(\mathbf{Q}_1)+\overline{\gamma}$. The declared equation follows from Lemma 4. 
\end{IEEEproof}

This theorem manifests that reducing $s(\mathbf{Q}_1)$ is equivalent to reducing $\rho(\mathbf{Q}_1+\overline{\gamma}\mathbf{E}_N)$. By Lemma 3, this can be achieved by diminishing some entries of $\mathbf{Q}_1$. For instance, deleting a number of key edges from $G_v$ helps suppress computer viruses.

By Theorem 4, enhancing $s(\mathbf{Q}_2)$ benefits the prevention of electronic infections. Similarly, we have the following result.

\begin{thm}
	Consider model (4) and let $\overline{\alpha} = \max_{1 \leq i \leq N}\alpha_i$. Then 
	$
	s(\mathbf{Q}_2)=\rho(\mathbf{Q}_2+\overline{\alpha}\mathbf{E}_N)-\overline{\alpha}.
	$
\end{thm}

This theorem manifests that enhancing $s(\mathbf{Q}_2)$ is equivalent to enhancing $\rho(\mathbf{Q}_2+\overline{\alpha}\mathbf{E}_N)$. By Lemma 3, this can be achieved by enlarging some entries of $\mathbf{Q}_2$. For instance, adding a number of new edges in $G_p$ helps protect against digital infections.

\section{Accuracy of the linear SIPS model}

As was mentioned in Section 2, the exact SIPS model accurately captures the average interacting dynamics between viruses and patches. A question arises naturally: under what conditions does the linear SIPS model satisfactorily capture this dynamics? This section is intended to empirically give an answer to the question.


For the comparison purpose, we need to numerically solve a number of specific $3^N$-dimensional exact SIPS models, because their closed-form solutions are far beyond our reach. Based on the standard Gillespie algorithm for numerically solving continuous-time Markov chain models \cite{Gillespie1977}, we develop a numerical algorithm for solving the exact SIPS model, which is formulated as the following EXACT-SIPS algorithm. In our experiments, the number of sample paths in the algorithm is set to be $M = 10^{4}$.

\begin{algorithm}[!t]
	\caption{EXACT-SIPS}
	\label{alg1}
	\hspace*{0.02in} {\bf Input} \\
	$\mathbf{M}_{\beta} = \left(\beta_{ij}\right)_{N \times N}$;
	$\mathbf{M}_{\delta_1} = \left(\delta^{(1)}_{ij}\right)_{N \times N}$;
	$\mathbf{M}_{\delta_2} = \left(\delta^{(2)}_{ij}\right)_{N \times N}$; \\
	$\mathbf{D}_{\gamma} = diag\left(\gamma_i\right)$;
	$\mathbf{D}_{\alpha} = diag\left(\alpha_i\right)$; \\
	$M$: number of sample paths; \\
	$\mathbf{x}_0$: common initial state of all sample paths; \\
	$T$: common terminating time of all sample paths; \\
	\hspace*{0.02in} {\bf Output} \\
	$\overline{\mathbf{x}}(t), 0 \leq t \leq T$: averaged sample paths; \\ $\overline{I}(t), 0 \leq t \leq T$: averaged fraction of infected nodes over all sample paths; \\
	$\overline{P}(t), 0 \leq t \leq T$: averaged fraction of patched nodes over all sample paths.
	\begin{algorithmic}[1]
		
		\FOR{$m := 1$ to $M$}
		\STATE $t := 0; \quad \mathbf{x}_m(0) := \mathbf{x}_0$;
		\WHILE{$t\leq T$}
		\STATE // $\lambda_i$ is the total rate of state transitions of node $i$;
		\STATE calculate $\lambda_i$, $1 \leq i \leq N$; 
		\STATE // $\lambda_{tot}$ is the total rate of state transitions of all nodes;
		\STATE calculate $\lambda_{tot} = \sum_{i=1}^N \lambda_i$;
		\STATE // $\Delta$ is the time interval from the current state transition to the next state transition, which follows the exponential distribution with mean $\frac{1}{\lambda_{tot}}$;
		\STATE determine $\Delta$;
		\IF{$t + \Delta > T$} 
		\STATE $\mathbf{x}_m(s) := \mathbf{x}_m(t), t < s \leq T$;          
		\ELSE
		\STATE // $j$ is the node whose state transits at time $t + \Delta$, which follows the discrete distribution $\Pr\{\text{node } i \}=\frac{\lambda_i}{\lambda_{tot}}$, $1 \leq i \leq N$;
		\STATE choose $j$;
		\STATE // $k$ is the state to which node $j$ transits at time $t + \Delta$, which follows the discrete distribution $\Pr\{\text{state } l\}=\frac{\text{rate at which node } j \text{ transits to state } l}{\lambda_j}$;
		\STATE choose $k$;
		\STATE $\mathbf{x}_m(s) := \mathbf{x}_m(t), t < s < t + \Delta$;
		\STATE $x_{mj}(t+\Delta) := k$;
		\STATE $x_{mi}(t+\Delta) := x_i(t), 1 \leq i \leq N, i \neq j$;
		\STATE $t := t+ \Delta$;
		\ENDIF
		\ENDWHILE{}
		\ENDFOR
		\STATE $\overline{\mathbf{x}}(t) := \frac{\sum_{m=1}^M \mathbf{x}_m(t)}{M}, 0 \leq t \leq T$;
		\STATE $\overline{I}(t) : = \frac{\sum_{m=1}^M \sum_{i=1}^N 1_{\{x_{mi}(t) = 1\}}}{MN}, 0 \leq t \leq T$;
		\STATE $\overline{P}(t) : = \frac{\sum_{m=1}^M \sum_{i=1}^N 1_{\{x_{mi}(t) = 2\}}}{MN}, 0 \leq t \leq T$.
	\end{algorithmic}
\end{algorithm}

As the initial state in the following experiments, a randomly chosen node is set as infected, a randomly chosen node is set as patched, and set all the remaining nodes as susceptible. 


Scale-free networks are a large class of networks having widespread applications \cite{Barabasi1999}. Take a randomly generated scale-free network with 100 nodes as the virus-spreading network as well as the patch-distributing network. By taking 1025 random combinations of the parameters, we get 1025 pairs of linear and exact SIPS models, among which there are 36 pairs that satisfy the condition in Theorem 2, 39 pairs that satisfy the condition in Theorem 3, 156 pairs that satisfy the condition in Theorem 4, and 794 pairs that satisfy neither of them. After observations, we find that for each of the four collections of pairs, the way that the dynamics of a nonlinear SIPS model deviates from that of the corresponding exact SIPS model is similar. Figs. 4-7 give the comparison results of two pairs for each collection, respectively.

\begin{figure}[!t]
	\centering
	\includegraphics[width=3in]{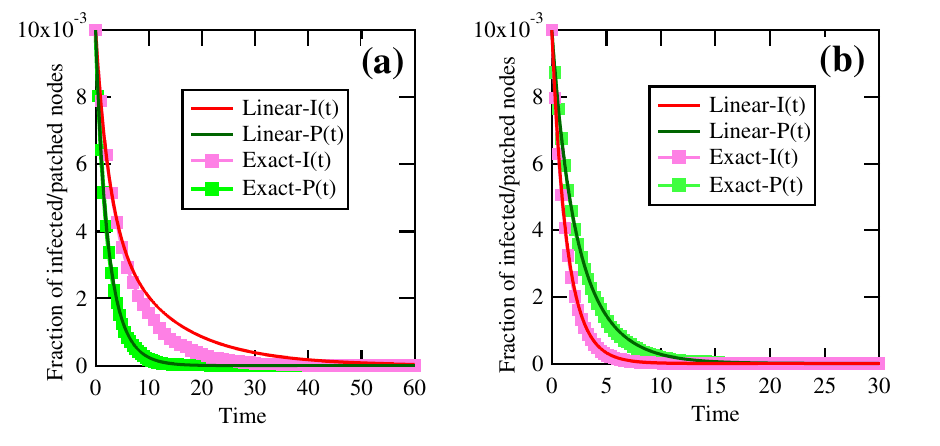}
	\caption{Comparison results for two pairs in the first collection.}
\end{figure}

\begin{figure}[!t]
	\centering
	\includegraphics[width=3in]{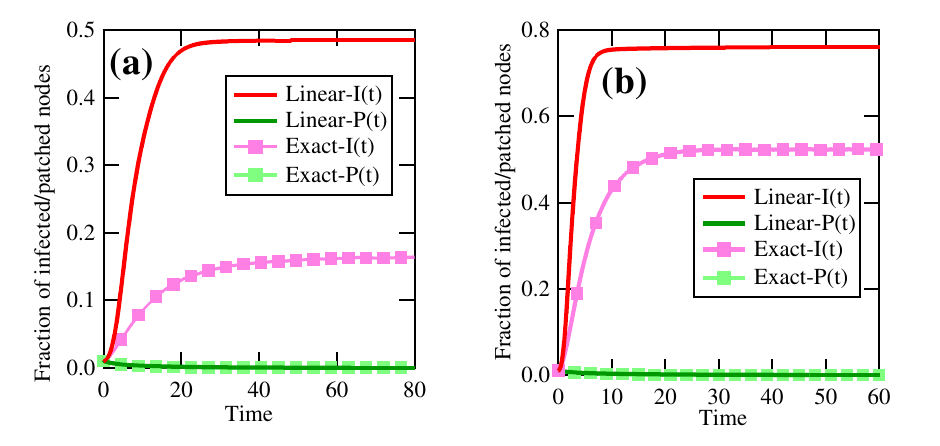}
	\caption{Comparison results for two pairs in the second collection.}
\end{figure}

\begin{figure}[!t]
	\centering
	\includegraphics[width=3in]{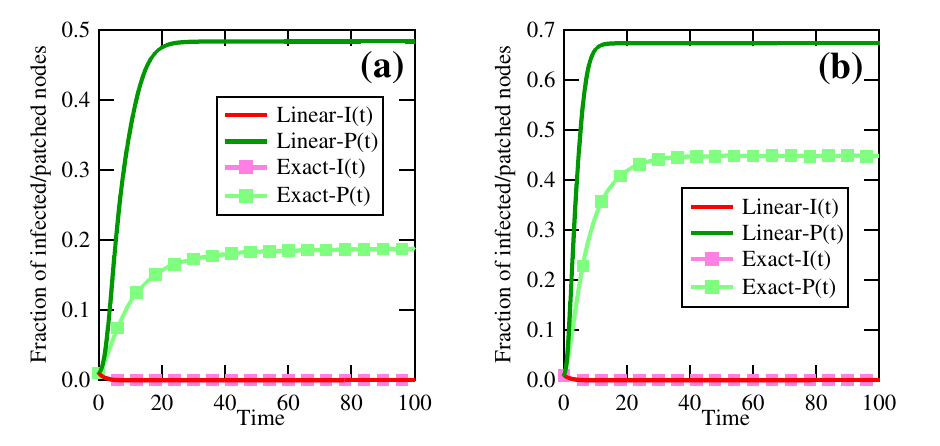}
	\caption{Comparison results for two pairs in the third collection.}
\end{figure}

\begin{figure}[!t]
	\centering
	\includegraphics[width=3in]{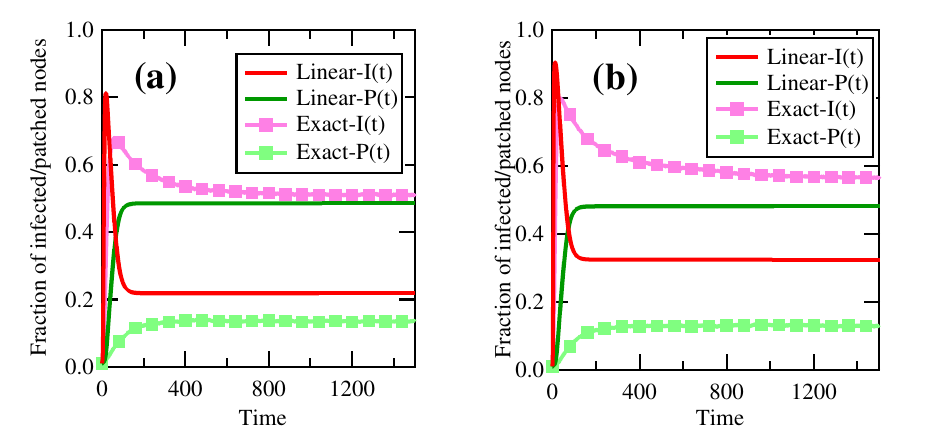}
	\caption{Comparison results for two pairs in the fourth collection.}
\end{figure}


Small-world networks are another large class of networks having widespread applications \cite{Watts1998}. Take a randomly generated small-word network with 100 nodes as the virus-spreading network and the patch-distributing network. By taking 1025 random combinations of the parameters, we get 1025 pairs of linear and exact SIPS models, among which there are 64 pairs that satisfy the condition in Theorem 2, 72 pairs that satisfy the condition in Theorem 3, 192 pairs that satisfy the condition in Theorem 4, and 697 pairs that satisfy neither of them. After observations, we find that for each of the four collections of pairs, the way that the dynamics of a nonlinear SIPS model deviates from that of the corresponding exact SIPS model is similar. Figs. 8-11 give the comparison results of two pairs for each collection, respectively.

\begin{figure}[!t]
	\centering
	\includegraphics[width=3in]{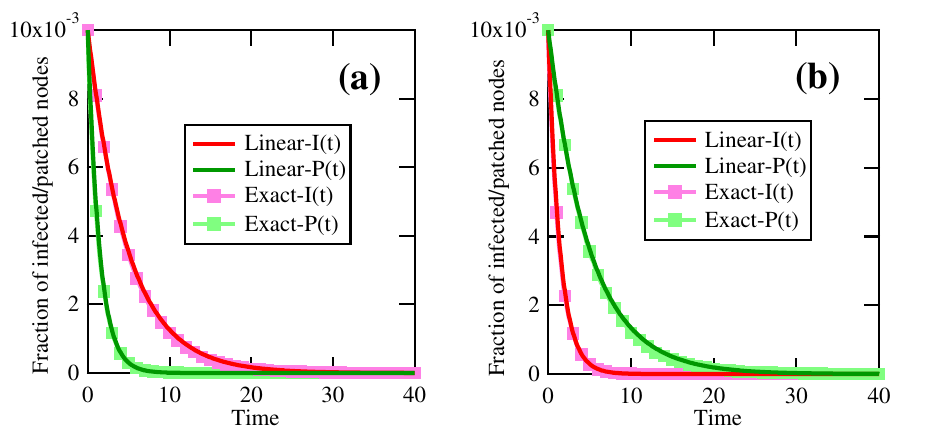}
	\caption{Comparison results for two pairs in the first collection.}
\end{figure}

\begin{figure}[!t]
	\centering
	\includegraphics[width=3in]{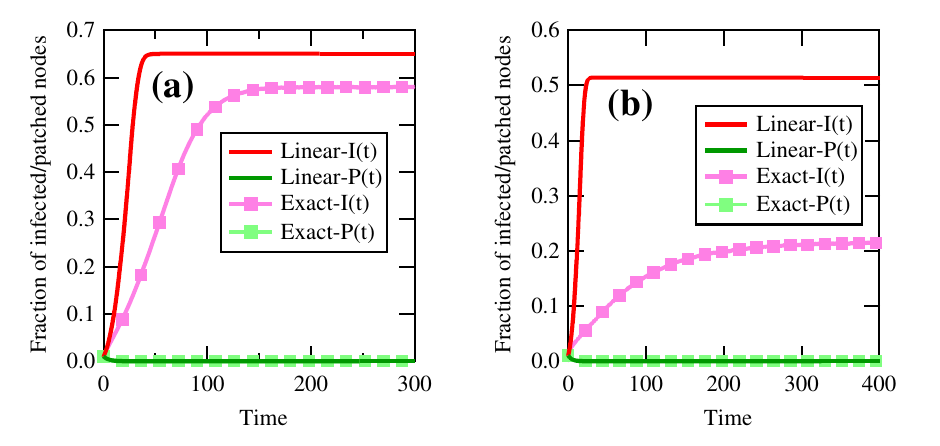}
	\caption{Comparison results for two pairs in the second collection.}
\end{figure}

\begin{figure}[!t]
	\centering
	\includegraphics[width=3in]{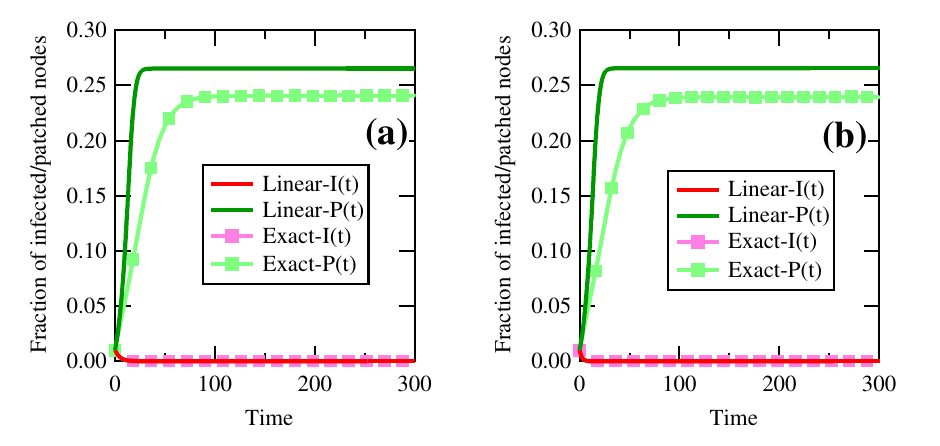}
	\caption{Comparison results for two pairs in the third collection.}
\end{figure}

\begin{figure}[!t]
	\centering
	\includegraphics[width=3in]{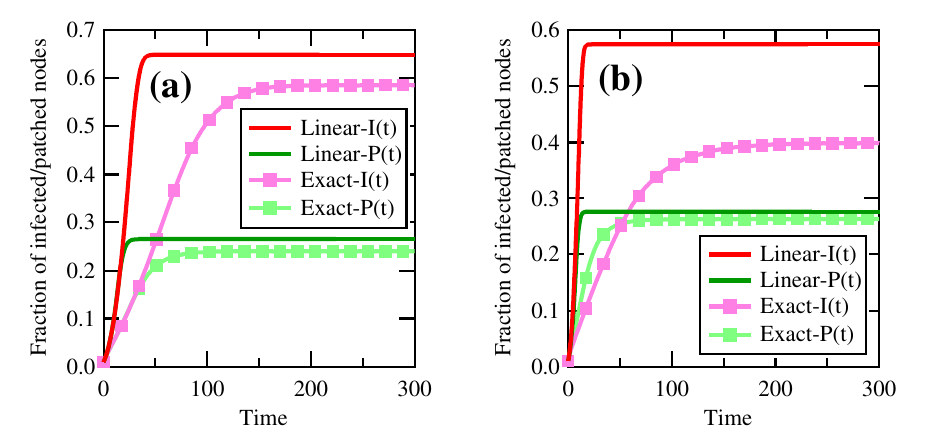}
	\caption{Comparison results for two pairs in the fourth collection.}
\end{figure}


The following conclusions can be drawn from the previous examples.

\begin{enumerate}
	\item[(a)] If $s(\mathbf{Q}_1) \leq 0$, then the linear SIPS model accurately captures the average extinction process of the viruses.
	\item[(b)] If $s(\mathbf{Q}_1) \leq 0$ and $s(\mathbf{Q}_2) \leq 0$, or if $s(\mathbf{Q}_1) > 0$ and $s(\mathbf{Q}_4) \leq 0$, then the linear SIPS model accurately captures the average extinction process of the patches.
	\item[(c)] If $s(\mathbf{Q}_1) > 0$ and $s(\mathbf{Q}_4) \leq 0$, or if $\mathbf{g} = \mathbf{h}$, $s(\mathbf{Q}_2) > 0$, and $ 
	s\left(\mathbf{Q}_1- (\mathbf{Q}_1+\mathbf{D}_{\gamma})diag\mathbf{P}^*-diag\mathbf{g}(\mathbf{P}^*)\right)>0, $
	then the linear SIPS model is not capable of accurately capturing the average evolution process of the viruses.
	\item[(d)] If $s(\mathbf{Q}_1) \leq 0$ and $s(\mathbf{Q}_2) > 0$, or if $\mathbf{g} = \mathbf{h}$, $s(\mathbf{Q}_2) > 0$, and $ 
	s\left(\mathbf{Q}_1- (\mathbf{Q}_1+\mathbf{D}_{\gamma})diag\mathbf{P}^*-diag\mathbf{g}(\mathbf{P}^*)\right)>0, $ 
	then the linear SIPS model is not capable of accurately capturing the average evolution process of the patches.
\end{enumerate}

In the case where the linear SIPS model works well, it can be employed to quickly predict the average evolution process of the viruses or/and patches in a population. 

In the case where the linear SIPS model doesn't work well, we have to resort to a generic SIPS model with nonlinear infecting/patching rates to achieve the goal of accurate prediction. In this case, it is critical to develop a methodology of choosing proper infecting/patching rates.

\section{Concluding remarks}

To accurately evaluate the performance of the decentralized patch distribution mechanism, a new virus-patch interacting model, which is known as the generic SIPS model, has been proposed. This model subsumes the linear SIPS model. Under the generic SIPS model, a set of criteria for the final extinction or/and long-term survival of viruses or/and patches have been presented. Some conditions for the linear SIPS model to accurately capture the average dynamics of the virus-patch interaction have been empirically found. We believe this work takes an important step towards the accurate assessment of the decentralized patch distribution mechanism.

Towards this direction, there are a number of key problems that are worth study. The obtained results should be applied to some specific networks such as the scale-free networks and the small-world networks. The combined influence of the virus-propagating network and the patch-distributing network on the performance of the decentralized patch distribution scheme should be investigated. In the case where the linear SIPS model doesn't work well, it is critical to develop a methodology of choosing proper infecting/patching rates so that the resulting generic SIPS model accurately predicts the average dynamics of the virus-patch interaction. The network bandwidth consumed by the decentralized patch distribution mechanism should be taken into consideration. Finally, the theory developed in this work may be applied to other areas such as the positive/negative information interactions \cite{WenS2015, Haghighi2016}.

\section*{Acknowledgment}

This work is supported by National Natural Science Foundation of China (Grant No. 61572006).

\ifCLASSOPTIONcaptionsoff
  \newpage
\fi

\end{document}